\pdfoutput=1
\documentclass{amsproc}
\usepackage{a4wide}
\usepackage{bookmark}

\usepackage{}

\newtheorem{theorem}{Theorem}[section]

\theoremstyle{definition}

\theoremstyle{remark}
\newtheorem{remark}[theorem]{Remark}

\newtheorem{conjecture}[theorem]{Conjecture}
\numberwithin{equation}{section}

\begin{document}

 \title[On the sum of linear coefficients]{On the Sum of Linear Coefficients of a Boolean Valued Function}


\author{Sumit Kumar Jha}
\address{Center for Security, Theory, and Algorithmic Research\\ 
International Institute of Information Technology, 
Hyderabad, India}
\curraddr{}
\email{kumarjha.sumit@research.iiit.ac.in}
\thanks{}


\date{}

\begin{abstract}
Let $f:\{-1,1\}^{n}\rightarrow \{-1,1\}$ be a Boolean valued function having total degree $d$. Then a conjecture due to Servedio and Gopalan asserts that $\sum_{i=1}^{n}\widehat{f}(i)\leq \sum_{j=1}^{d}\widehat{\text{Maj}}_{d}(j)$ where $\text{Maj}_{d}$ is the majority function on $d$ bits. Here we give some alternative formalisms of this conjecture involving the discrete derivative operators on $f$.
\end{abstract}

\maketitle
\section{Introduction: Fourier Analysis of Boolean Functions}
We are concerned with \emph{Boolean-valued} functions $f:\{-1,1\}^{n}\rightarrow \{-1,1\}$ which form a subset of \emph{Boolean} functions $f:\{-1,1\}^{n}\rightarrow \mathbb{R}$. Every Boolean function $f:\{-1,1\}^{n}\rightarrow \mathbb{R}$ has a unique \emph{Fourier expansion} given by
$$f(x)=\sum_{S\subseteq [n]}\widehat{f}(S)\prod_{i\in S}x_{i},$$
where the real numbers $\widehat{f}(S)$ are the \emph{Fourier coefficients} of $f$ given by the formula
$$ \widehat{f}(S)=\textbf{E}\left[f(x)\prod_{i\in S}x_{i}\right].$$
(Here and everywhere else in the paper, the expectation $\textbf{E}[\cdot]$ is with respect to the uniform probability distribution on $\{-1,1\}^{n}$.) The \emph{Parseval's identity}  is the fact that $\sum_{S\subseteq [n]}\widehat{f}(S)^{2}=\textbf{E}[f(x)^{2}].$ In particular, if $f$ is boolean-valued then this implies that $\sum_{S\subseteq [n]}\widehat{f}(S)^{2}=1$.\par
Given $f:\{-1,1\}^{n}\rightarrow \mathbb{R}$ and $i\in [n]$, we define the \emph{discrete derivative} $\partial_{i}f:\{-1,1\}^{n}\rightarrow \mathbb{R}$ by
$$\partial_{i}f(x)=\frac{f(x_{1},x_{2},\cdots ,1,x_{i+1},\cdots ,x_{n})-f(x_{1},x_{2},\cdots ,-1,x_{i+1},\cdots ,x_{n})}{2}.$$
The \emph{influence} of the $i$th coordinate on $f$ is defined by
$$\textbf{Inf}_{i}[f]=\textbf{E}[(\partial_{i}f)^{2}]=\sum_{S\ni i}\widehat{f}(S)^{2}.$$
In the particular case when $f$ is Boolean-valued, the derivative $\partial_{i}f$ is $\{-1,0,1\}$-valued. The \emph{total influence} of $f$ is 
$$\textbf{Inf}[f]=\sum_{i=1}^{n}\textbf{Inf}_{i}[f]=\sum_{S\subseteq [n]}|S|\cdot \widehat{f}(S)^{2}.$$
The \emph{total degree} of $f$ is defined by
$$\text{deg}(f)=\max\{|S|:\widehat{f}(S)\neq 0\}.$$
Note that for a Boolean-valued function $f:\{-1,1\}^{n}\rightarrow \{-1,1\}$
$$\textbf{Inf}[f]=\sum_{S\subseteq [n]}|S|\cdot \widehat{f}(S)^{2}\leq \text{deg}(f)\cdot \sum_{S\subseteq [n]}\widehat{f}(S)^{2}=\text{deg}(f),$$
where we used the Parseval's identity to deduce the last step.\par 
The linear coefficients of $f$ are the $n$ coefficients $\widehat{f}(\{1\}), \widehat{f}(\{2\}),\cdots,\widehat{f}(\{n\})$, and hereon we omit writing the curly braces inside to denote them. We are concerned with the following conjecture here:
\begin{conjecture}[Gopalan-Servedio \cite{open}]
\label{conj1}
Let $f:\{-1,1\}^{n}\rightarrow \{-1,1\}$ have total degree $d$. Let $\text{Maj}_{d}:\{-1,1\}^{d}\rightarrow \{-1,1\}$ be defined by $\text{Maj}_{d}(x)=\text{sgn}(x_{1}+x_{2}+\cdots+x_{d})$ with $\text{sgn}(0)=-1$. Then
$$\normalfont \sum_{i=1}^{n}\widehat{f}(i)\leq \sum_{j=1}^{d}\widehat{\text{Maj}}_{d}(j).$$
\end{conjecture}
\begin{remark}
The conjecture is trivial when $\text{deg}(f)=n$ since
$$\sum_{i=1}^{n}\widehat{f}(i)=2^{-n}\sum_{x\in \{-1,1\}^{n}}f(x)\cdot (x_{1}+x_{2}+\cdots+x_{n})\leq 2^{-n}\sum_{x\in \{-1,1\}^{n}}|x_{1}+x_{2}+\cdots+x_{n}|=\sum_{j=1}^{n}\widehat{\text{Maj}}_{n}(j).$$
\end{remark}
\section{Alternative Formalisms of the conjecture}
\begin{theorem}
The Conjecture \ref{conj1} is equivalent to each of the following inequalities
$$\normalfont \textbf{Inf}[f]-\textbf{Inf}[\text{Maj}_{d}] \leq 2 \cdot \sum_{k=1}^{n}\textbf{Pr}\left[\partial_{k} f=-1\right],$$
$$\normalfont \sum_{k=1}^{n}\left(\textbf{Pr}\left[\partial_{k} f=1\right]-\textbf{Pr}\left[\partial_{k} \text{Maj}_d=1\right]\right)\leq \sum_{j=1}^{n}\textbf{Pr}\left[\partial_{j} f=-1\right],$$
$$\normalfont 2\cdot \sum_{k=1}^{n}\textbf{Pr}\left[\partial_{k} f=1\right]\leq \textbf{Inf}[f]+\textbf{Inf}[\text{Maj}_d].$$
\end{theorem}
\begin{proof}
Let 
$$ f(x)=\sum_{\substack{S\subseteq [n]}}\widehat{f}(S)\prod_{i\in S}x_{i},$$
and
$$\partial_{i}f(y):=\frac{f(y_{1},y_{2},\cdots ,1,y_{i+1},\cdots ,y_{n})-f(y_{1},y_{2},\cdots ,-1,y_{i+1},\cdots ,y_{n})}{2}$$
for all $y\in \{-1,1\}^{n-1}$ which takes values $+1,0,-1$. We have
$$ \textbf{Pr}[\partial_{i}f=0]=2^{-(n-1)}\sum_{y\in \{-1,1\}^{n-1}}(1+\partial_{i}f(y))\cdot (1-\partial_{i}f(y))=1-\textbf{Inf}_{i}[f],$$
$$\textbf{Pr}[\partial_{i}f=1]=2^{-(n-1)}\sum_{y\in \{-1,1\}^{n-1}}\frac{(\partial_{i}f(y)+1)\cdot \partial_{i}f(y)}{2}=\frac{1}{2}(\textbf{Inf}_{i}[f]+\widehat{f}(i)),$$
$$\textbf{Pr}[\partial_{i}f=-1]=2^{-(n-1)}\sum_{y\in \{-1,1\}^{n-1}}\frac{(\partial_{i}f(y)-1)\cdot \partial_{i}f(y)}{2}=\frac{1}{2}(\textbf{Inf}_{i}[f]-\widehat{f}(i)).$$
where we used the fact that $\widehat{f}(i)=\textbf{E}[(\partial_{i}f)].$ Each of the assertions follow easily from the above equations.
\end{proof}
\bibliographystyle{amsplain}
\bibliography{sample.bib}
\end{document}